\documentclass[twoside,english]{article}
\usepackage{babel}

\usepackage[accepted]{aistats2017}
\usepackage[T1]{fontenc}
\usepackage[latin9]{inputenc}
\usepackage{amsthm}
\usepackage{amsmath}
\usepackage{amssymb}
\usepackage{graphicx}
\usepackage{dsfont}
\usepackage{cite}
\usepackage{amsmath}
\usepackage{array}
\usepackage{multirow}
\usepackage{caption}
\usepackage{color}

\usepackage{multicol}

\theoremstyle{plain}

\theoremstyle{plain}

\theoremstyle{plain}
\newtheorem{lem}{\protect\lemmaname}
\theoremstyle{plain}
\newtheorem{thm}{\protect\theoremname}
\theoremstyle{plain}
  
\theoremstyle{definition}

\theoremstyle{definition}

\theoremstyle{definition}

\makeatother
  
\usepackage{babel} 

\providecommand{\claimname}{Claim}
\providecommand{\lemmaname}{Lemma}
\providecommand{\propositionname}{Proposition}
\providecommand{\theoremname}{Theorem}
\providecommand{\corollaryname}{Corollary} 
\providecommand{\definitionname}{Definition}
\providecommand{\assumptionname}{Assumption}
\providecommand{\remarkname}{Remark}

\newcommand{\overbar}[1]{\mkern 1.25mu\overline{\mkern-1.25mu#1\mkern-0.25mu}\mkern 0.25mu}

\newcommand{\lammin}{\lambda_{\mathrm{min}}}
\newcommand{\lammax}{\lambda_{\mathrm{max}}}
\newcommand{\missing}{\ast}

\newcommand{\bSigmatil}{\widetilde{\mathbf{\Sigma}}}
\newcommand{\mtil}{\widetilde{m}}
\newcommand{\Ptil}{\widetilde{P}}
\newcommand{\dmin}{d_{\mathrm{min}}}
\newcommand{\davg}{d_{\mathrm{avg}}}
\newcommand{\dmaxbar}{\overline{d}_{\mathrm{max}}}
\newcommand{\taumin}{\tau_{\mathrm{min}}}
\newcommand{\taumax}{\tau_{\mathrm{max}}}
\newcommand{\ntil}{\widetilde{n}}

\newcommand{\mmin}{m_{\mathrm{min}}}
\newcommand{\mmax}{m_{\mathrm{max}}}

\newcommand{\dmax}{d_{\mathrm{max}}}

\newcommand{\Tr}{\mathrm{Tr}}

\newcommand{\pebar}{\overbar{P}_{\mathrm{e}}}

\newcommand{\pe}{P_{\mathrm{e}}}

\newcommand{\Xv}{\mathbf{X}}

\newcommand{\Gc}{\mathcal{G}}

\newcommand{\Tc}{\mathcal{T}}

\newcommand{\EE}{\mathbb{E}}
\newcommand{\PP}{\mathbb{P}}
\newcommand{\RR}{\mathbb{R}}

\newcommand{\Iv}{\mathbf{I}}

\newcommand{\bzero}{\mathbf{0}}

\newcommand{\bSigma}{\mathbf{\Sigma}}

\usepackage{float}
\usepackage{subfig}

\begin{document}

%
\runningtitle{{Lower Bounds on Active Learning for Graphical Model Selection}}

%

\twocolumn[

\aistatstitle{ Lower Bounds on Active Learning for Graphical Model Selection}
\runningauthor{Jonathan Scarlett and Volkan Cevher}
\runningtitle{{Lower Bounds on Active Learning for Graphical Model Selection}}
%
 
\aistatsauthor{ Jonathan Scarlett and Volkan Cevher }

\aistatsaddress{ 
   Laboratory for Information and Inference Systems (LIONS) \\ 
    \'Ecole Polytechnique F\'ed\'erale de Lausanne (EPFL) \\
    Email: \{jonathan.scarlett, volkan.cevher\}@epfl.ch } 
    
]

\begin{abstract}
    We consider the problem of estimating the underlying graph associated with a Markov random field, with the added twist that the decoding algorithm can iteratively choose which subsets of nodes to sample based on the previous samples, resulting in an active learning setting.  Considering both Ising and Gaussian models, we provide algorithm-independent lower bounds for high-probability recovery within the class of degree-bounded graphs.  Our main results are minimax lower bounds for the active setting that match the best known lower bounds for the passive setting, which in turn are known to be tight in several cases of interest.     Our analysis is based on Fano's inequality, along with novel mutual information bounds for the active learning setting, and the application of restricted graph ensembles.  While we consider ensembles that are similar or identical to those used in the passive setting, we require different analysis techniques, with a key challenge being bounding a mutual information quantity associated with observed subsets of nodes, as opposed to full observations.
\end{abstract}

\section{Introduction}

Graphical models are a widely-used tool for providing compact representations of the conditional independence relations between random variables, and arise in areas such as image processing \cite{Gem84}, statistical physics \cite{Gla63}, computational biology \cite{Dur98}, natural language processing \cite{Man99}, and social network analysis \cite{Was94}.  The problem of \emph{graphical model selection} consists of recovering the graph structure given a number of independent samples from the underlying distribution.  While this problem is NP-hard in general \cite{Chi96}, there exist a variety of methods guaranteeing exact recovery with high probability on \emph{restricted} graph classes, with a particularly common restriction being bounded degree.

Several variations of graphical model selection problems with \emph{active learning} have appeared in the literature.  In this paper, we adopt the formulation given in \cite{Das16}, in which the recovery algorithm may adaptively choose which nodes to sample, based on the previous samples.  The goal is to recover the underlying graph subject to a constraint on the total number of node observations.   As discussed in \cite{Das16}, this variation is of interest in several applications; for example, in sensor networks one may be able to choose which sensors to activate, rather than simultaneously activating every sensor at every time instant.  Only upper bounds were provided in \cite{Das16}, and the problem of finding lower bounds was left as an open problem.


\subsection{Contributions}

In this paper, we complement the work of \cite{Das16} by providing algorithm-independent lower bounds on active learning for graphical model selection.  Our main findings are summarized as follows:
\begin{enumerate}
    \item For both Ising models and Gaussian models, we provide lower bounds that essentially match the best known lower bounds for the passive setting \cite{San12,Wan10}, in terms of the minimax probability of error with respect to the class of bounded-degree graphs.  The passive learning bounds are known to be tight in several cases of interest, and our results show that \emph{active learning does not help significantly in the minimax sense} in such cases. 
    \item We provide a class of Gaussian graphical models where the average degree dictates the lower bounds as opposed to the maximal degree, and where we match upper bounds based on the average degree in \cite{Das16}.  Hence, we identify a graph class where the average degree is provably the fundamental quantity dictating the fundamental limits.  Moreover, we provide a class of Ising models where the maximal degree provably remains the key quantity dictating the performance, hence revealing that one cannot always improve the dependence from the maximal to the average degree.
\end{enumerate}
Our analysis uses a variation of Fano's inequality for the active learning setting, along with novel mutual information bounds proved using analogous techniques to those used in channel coding with noiseless feedback \cite{Cov01}.  We apply the resulting bound to a variety of restricted graph ensembles in which the graphs are difficult to distinguish from each other, with notable examples being (i) isolated edges that are difficult to detect; (ii) cliques with a single edge removed such that the removal is difficult to detect.  While the ensembles that we use are similar or identical to those used in the passive setting, analyzing them in the active setting requires new techniques, particularly for bounding a mutual information quantity associated with partial observations instead of full observations.

\subsection{Related Work} \label{sec:related}

In the same way that feedback often provides little or no gain in the capacity for channel coding \cite{Cov01}, it is often observed that active learning provides little or no gain in the information-theoretic sample complexity of inference and learning problems.  For example, in the compressive sensing problem, it has been shown that the improvement amounts to at most a logarithmic factor \cite{Ari13}.  For the group testing problem, under a broad range of scalings of the sparsity level, not even the constant factors improve \cite{Sca15b,Bal13}.  

On the other hand, active learning is known to strictly improve the sample complexity in several cases of interest \cite{Cas08,Hau11}.  Moreover, it should be noted that even when adaptivity does not help asymptotically in an information-theoretic sense, it can still help in the sense of leading to simpler and less computationally expensive algorithms, and also in improving the non-asymptotic performance \cite{Bal13,Pol11,Hau11}.

Active learning for graphical model selection has been studied in several contexts \cite{Mur01,Ton01,Das16}, the most relevant to ours being that of Dasarathy \emph{et al.} \cite{Das16}.  A general algorithm was proposed therein using abstract subroutines for neighborhood selection and neighborhood verification, and applications to the Gaussian setting revealed cases where the total number of node observations is improved from $O(\dmax p\log p)$ to $O((1+\dmaxbar) p\log p)$.  Here $\dmaxbar$ is the average of the node-wise maximal degree, where the latter is defined as the highest degree among a node and all its neighbors.  This quantity can be significantly smaller than $\dmax$, in which case the improvement in the sample complexity is substantial.

Information-theoretic lower bounds for the passive setting were given in \cite{San12,Bre08,Ana12,Tan13,Sha14,Sca16,Vat11} for the Ising model, and \cite{Wan10,Ana12a,Jog15} for the Gaussian model.  Let $\ntil$ be the sample complexity with respect to the number of $p$-dimensional observations.  The best minimax lower bounds for degree-bounded graphs are summarized as follows for the Ising model \cite{San12}:
\begin{equation}
    \ntil = \Omega\bigg(\max\bigg\{ \frac{\log p}{ \lambda\tanh\lambda }, \frac{e^{\lambda d} \log(pd) }{ \lambda d e^{\lambda} }, d\log\frac{p}{d} \bigg\}\bigg), \label{eq:Ising_prev}
\end{equation}
where $p$ is the number of nodes, $d$ is the maximal degree, and $\lambda$ is the inverse temperature of the Ising model (see Section \ref{sec:PRELIM} for precise definitions).  For the Gaussian model, the best known lower bounds for degree-bounded graphs are \cite{Wan10}
\begin{equation}
    \ntil = \Omega\bigg(\max\bigg\{ \frac{\log p}{ \tau^2 }, \frac{d\log\frac{p}{d}}{\log(1+d\tau)} \bigg\}\bigg), \label{eq:Gaussian_prev}
\end{equation}
where $\tau$ corresponds to the smallest allowed off-diagonal magnitude in the normalized inverse covariance matrix (see Section \ref{sec:PRELIM} for details). 

A wide range of polynomial-time algorithms have been proposed for the passive learning of graphical models; see \cite{Bre08,Ana12,Wu13,Jal11,Ray12,Bre14,Bre14a,Rav10,Vat11} for Ising models, and \cite{Rav11,Mei06,Yan14,Ana12} for Gaussian models.  The best performance bounds among these algorithms match those of \eqref{eq:Ising_prev}--\eqref{eq:Gaussian_prev} in several cases of interest, though there are other cases where gaps remain, or where the results are difficult to compare due to the differences in the underlying assumptions (e.g., additional coherence assumptions).

\subsection{Structure of the Paper}

In Section \ref{sec:PRELIM}, we formally define the Ising and Gaussian graphical models, and formulate the active learning problem.  Our main results are presented and discussed in Section \ref{sec:RESULTS}.  The proofs are given in Section \ref{sec:Fano} (Fano's inequality), Section \ref{sec:PROOF_ISING} (Ising model), and Section \ref{sec:PROOF_GAUSSIAN} (Gaussian model).  In Section \ref{sec:DISCUSSION}, we discuss the role of the average vs.~maximal degree, and we conclude our work in Section \ref{sec:CONCLUSION}.

\section{Active Learning for Graphical Model Selection} \label{sec:PRELIM}

\subsection{Preliminaries}

We consider a collection of $p$ random variables $(X_1,\dotsc,X_p)$ whose joint distribution is encoded by a graphical model $G=(V,E)$ with vertex set $V=\{1,\dotsc,p\}$ and undirected edge set $E$.  The elements of $V$ are referred to as \emph{nodes} or \emph{variables} interchangeably.  We use the standard terminology that the \emph{degree} of a node $i\in V$ is the number of edges in $E$ containing $i$, and that a \emph{clique} is a fully-connected subset of $V$ of cardinality at least two.

We consider two classes of joint probability distributions encoded by $G$, namely, Ising models and Gaussian models.  These are described as follows.

\textbf{Ising Model:}
In the ferromagnetic Ising model \cite{Isi25,Lau96}, each vertex is associated with a binary random variable $X_i \in \{-1,1\}$, and the corresponding joint distribution is described by the probability mass function
\begin{equation}
    P_{G}(x) = \frac{1}{Z}\exp\bigg( \lambda \sum_{(i,j)\in E} x_i x_j \bigg), \label{eq:Ising}
\end{equation}
where $Z$ is a normalizing constant called the partition function. Here $\lambda > 0$ is a parameter to the distribution, sometimes called the inverse temperature.

In the context of Ising model selection, we write $\Gc_{d}$ as $\Gc_{d,\lambda}$ to emphasize that the results depend on $\lambda$.  Although we let $\lambda$ be a constant here, our lower bounds remain valid in the minimax sense when one considers the larger class in which the edges have differing parameters $\{\lambda_{ij}\}$ in the range $[\lammin,\lammax]$, provided that $\lammin \le \lambda \le \lammax$.

\textbf{Gaussian Model:}
In the Gaussian graphical model \cite{Lau96}, each vertex is associated with a random variable $X_i \in \RR$, and the corresponding joint distribution is
\begin{equation}
    (X_1,\dotsc,X_p) \sim N(\bzero,\bSigma), \label{eq:Gaussian_dist}
\end{equation}
where $\bzero$ is the vector of zeros, and $\bSigma$ is a covariance matrix whose inverse $\bSigma^{-1}$ contains non-zeros only in the diagonal entries and the indices corresponding to pairs in $E$.  By the Hammersley-Clifford theorem \cite{Lau96}, this implies the Markov property for the graph, namely, that a given node is conditionally independent of the rest of the graph given its neighbors.  The joint density function corresponding to \eqref{eq:Gaussian_dist} is denoted by $P_G$, overloading the notation used above for the Ising model.

A typical restriction on the entries of $\Theta = \Sigma^{-1}$ is that $\frac{|\Theta_{ij}|}{\sqrt{\Theta_{ii}\Theta_{jj}}}$ is lower bounded by some constant $\tau > 0$ \cite{Wan10,Das16}.  We consider the simplest special case of this in which the lower bound always holds with equality:
\begin{equation}
    \Theta_{ij} = \begin{cases} 1 & i = j \\ \pm\tau & (i,j) \in E \\ 0 & \mathrm{otherwise}. \end{cases} \label{eq:Theta_entries}
\end{equation}
We write $\Gc_{d}$ as $\Gc_{d,\tau}$ to emphasize that the results depend on $\tau$.  Similarly to the Ising model, our lower bounds remain valid in the minimax case when we consider the larger class with $\frac{|\Theta_{ij}|}{\sqrt{\Theta_{ii}\Theta_{jj}}} \in [\taumin,\taumax]$ with $\taumin \le \tau \le \taumax$.  

\begin{figure}
    \begin{centering}
        \includegraphics[width=0.95\columnwidth]{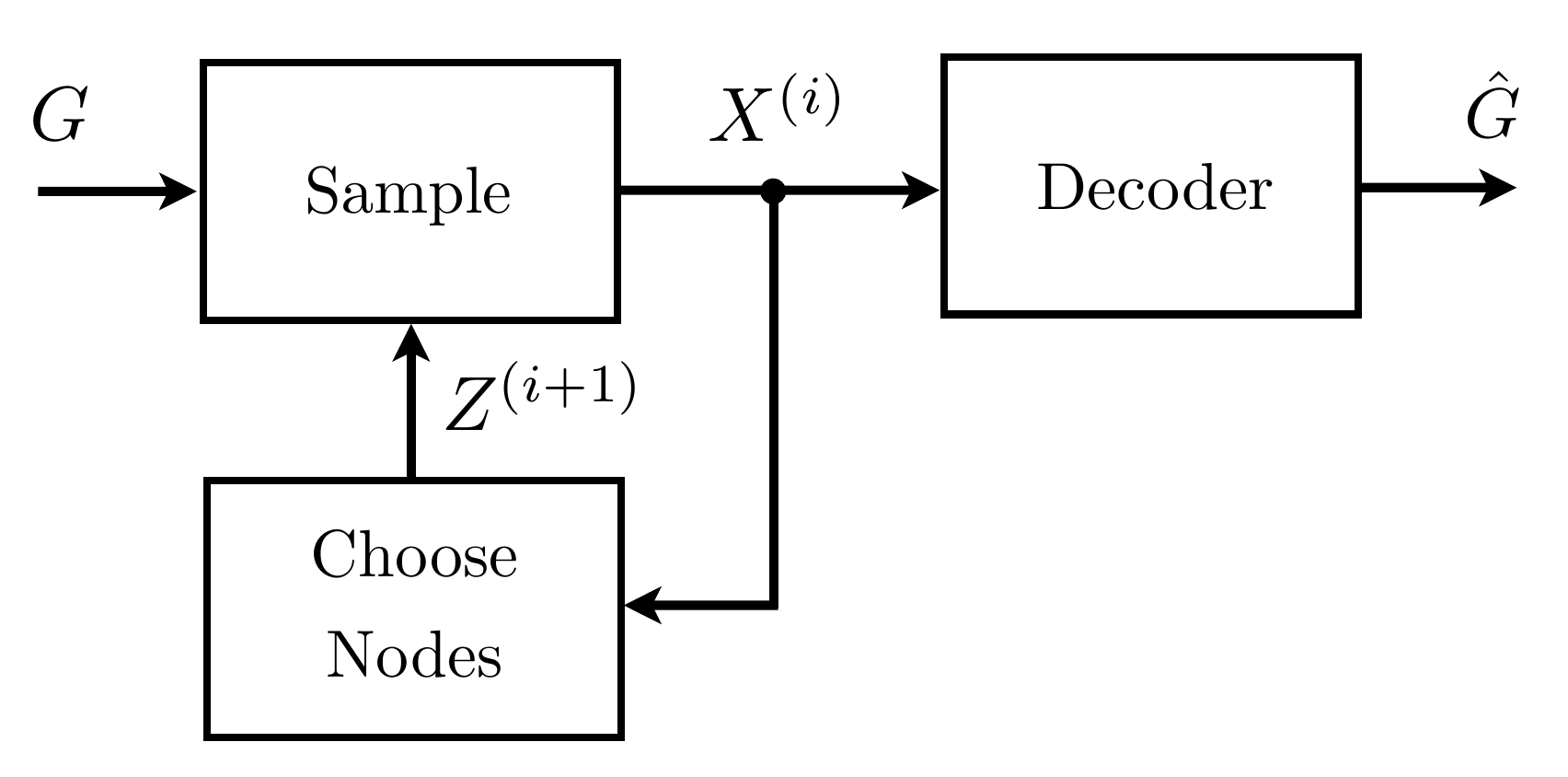}
        \par
    \end{centering}
    
    \caption{Illustration of the active learning problem for graphical model selection. \label{fig:problem_diagram}}
\end{figure}

\subsection{Problem Statement}

The problem of graphical model selection with active learning proceeds in rounds $i=1,2,\dotsc$, as illustrated in Figure \ref{fig:problem_diagram}.  In the $i$-th round, the algorithm selects a subset of $V$ to observe, encoded by a binary vector $Z^{(i)} \in \{0,1\}^p$ equaling one for observed nodes and zero for non-observed nodes.  The resulting \emph{sample} (or \emph{observation}) is a $p$-dimensional vector $X^{(i)}$ such that:
\begin{itemize}
    \item The joint distribution of the entries of $X^{(i)}$, corresponding to the entries where $Z^{(i)}$ is one, coincide with the corresponding joint distribution of the vector $(X_1,\dotsc,X_p) \sim P_G$, with independence between rounds;
    \item The values of the entries of $X^{(i)}$, corresponding to the entries where $Z^{(i)}$ is zero, are deterministically given by $\missing$, a symbol indicating that the node was not observed.
\end{itemize}

For convenience, we let $N$ denote the maximum possible number of active learning rounds (e.g., we can simply set $N=n$), and use the convention that for values of $i$ beyond the actual (possibly random) final round, $X^{(i)} = (\missing,\dotsc,\missing)$.  Letting $|Z^{(i)}|$ denote the number of entries where $Z^{(i)}$ is one, we refer to $\sum_{i=1}^N |Z^{(i)}|$ as the \emph{total number of node observations} used throughout the course of the algorithm, and we impose an upper bound on its maximum allowed value, denoted by $n$.  Note that this differs from the quantity $\ntil$ in \eqref{eq:Ising_prev}--\eqref{eq:Gaussian_prev} by a factor of $p$.

After the final round, the algorithm constructs an estimate $\hat{G}$ of $G$, and the error probability is given by
\begin{equation}
    \pe(G) := \PP[\hat{G} \ne G]. \label{eq:pe_G}
\end{equation}
We consider the class $\Gc_{d}$ of degree-bounded graphs, in which all nodes have degree at most $d$.  Specifically, we are interested in bounds on the minimax (worst-case) error probability for graphs in this class:
\begin{equation}
    \pe := \max_{G \in \Gc_d} \PP[\hat{G} \ne G], \label{eq:pe}
\end{equation}
where the dependence on the total number of node samples $n$ is kept implicit.  Note that when we consider the Gaussian setting, the maximum in \eqref{eq:pe_G} is not only over the graph $G$, but also implicitly over the signs ($+1$ or $-1$) in the second case of \eqref{eq:Theta_entries}.

We are interested in characterizing the {\em sample complexity}, meaning the required number of node observations $n$ needed in order to achieve $\pe \le \delta$ for some target error probability $\delta > 0$.

\section{Main Results} \label{sec:RESULTS}

In this section, we state and discuss our main results, namely, minimax lower bounds on the sample complexity for $\Gc_d$.  We note that the proofs are based on graph ensembles in which the maximal degree and average degree are approximately equal; however, in Section \ref{sec:DISCUSSION}, we discuss variations of these ensembles in which these two notions differ significantly.


\subsection{Ising Model}

\begin{thm} \label{thm:Ising}
    For Ising graphical models with $\lambda d \ge 1$, in order to recover any graph in $\Gc_{d,\lambda}$ with probability at least $1-\delta$, it is necessary that the total number of node observations, $n$, satisfies
    \begin{multline}
        n \ge \max \Bigg\{ \frac{2p\log p}{ \lambda \tanh \lambda }, \frac{e^{\lambda d} \log(pd) }{ 2\lambda d e^{\lambda}},  \frac{pd\log\frac{p}{8d} }{4\log 2}\Bigg\} \\ \times (1-\delta-o(1)). \label{eq:Ising_final}
    \end{multline}
\end{thm}
\begin{proof}
    See Section \ref{sec:PROOF_ISING}.
\end{proof}

The second bound in \eqref{eq:Ising_final} reveals that the sample complexity is very large when $\lambda d \to \infty$ at a rate that is not too slow, due to the exponential term $e^{\lambda d}$.  On the other hand, when $\lambda = O\big(\frac{1}{d}\big)$, the first bound gives a sample complexity of $\Omega(d^2 p\log p)$, since $\tanh\lambda = O(\lambda)$ as $\lambda \to 0$.  Finally, in any case, the third bound gives $n = \Omega\big(pd\log\frac{p}{d}\big)$.  These observations coincide with those for the lower bounds on passive learning in \cite{San12} (see \eqref{eq:Ising_prev} with $\ntil = np$), suggesting that active learning does not help much in the minimax sense for $\Gc_{d,\lambda}$.  Note that compared to \cite{San12}, we lose a factor of $p$ in the second bound, but this factor is insignificant compared to $e^{\lambda d}$ provided that $\lambda d \gg \log p$.

\subsection{Gaussian Model}

\begin{thm} \label{thm:Gaussian}
    For Gaussian graphical models with $d=o(p)$, in order to recover any graph in $\Gc_{d,\tau}$ with probability at least $1-\delta$, it is necessary that the total number of node observations, $n$, satisfies
    \begin{multline}
        n \ge \max \Bigg\{ \frac{4 p\log p}{ \log\frac{1}{1-\tau^2} }, \frac{2pd\log \frac{p}{d}}{ \log\Big(1 + \big((d+1)\frac{\tau}{1-\tau}\big)^2\Big) } \Bigg\} \\ \times (1-\delta-o(1)). \label{eq:Gaussian_final}
    \end{multline}
\end{thm}
\begin{proof}
    See Section \ref{sec:PROOF_GAUSSIAN}.
\end{proof}

When $\tau = o(1)$, the first bound behaves as $\Omega\big( \frac{1}{\tau^2} p \log p \big)$, whereas when $\tau$ is a constant, the second bound behaves as $\Omega\big( \frac{1}{\log d} \cdot pd\log\frac{p}{d} \big)$.  Both of these scaling laws are identical to the necessary conditions for passive learning in \cite{Wan10} (see \eqref{eq:Gaussian_prev} with $\ntil = np$), again suggesting that active learning does not help much in the minimax sense for $\Gc_{d,\tau}$. 

While the above findings indicate that active learning does not help much in the minimax sense for $\Gc_d$, we discuss a more restricted class of graphs in Section \ref{sec:DISCUSSION} for which active learning helps when $\tau$ is a constant.  Specifically, similarly to the upper bound in \cite{Das16}, the linear dependence on the maximal degree $d$ in the second term of \eqref{eq:Gaussian_final} is improved to the \emph{average} degree.
 
 \section{Proofs of Main Results}
 
\subsection{Fano's Inequality for Active Learning} \label{sec:Fano}

We first apply Fano's inequality \cite{Cov01} along with a novel mutual information bound for active learning in graphical model selection.  The proof bears some resemblance to that of the converse bound for channel coding with noiseless feedback \cite[Sec.~7.12]{Cov01}.  

For $z \in \{0,1\}^p$, we let $G(z)$ denote the subgraph of $G$ obtained by keeping only the nodes corresponding to entries where $z$ equals one, and denote the resulting joint distribution by $P_{G(z)}$.  More generally, for a joint distribution $Q$ on $p$ random variables labeled $\{1,\cdots,p\}$, we let $Q_{(z)}$ denote the joint marginal distribution corresponding to the entries where $z$ is one.

In the following lemma, we let $G$ be uniformly random on some subset of $\Gc_d$, and define the \emph{average} error probability
\begin{equation}
    \pebar := \PP[ \hat{G} \ne G ] = \EE[\pe(G)],
\end{equation}
where in contrast with \eqref{eq:pe}, the probability is now additionally over $G$.  Clearly any lower bound on the sample complexity for achieving $\pebar \le \delta$ implies the same lower bound for achieving $\pe \le \delta$, since $\pe$ is defined with respect to the worst case.  

\begin{lem} \label{lem:Fano}
    Let $G$ be uniform over a restricted graph class $\Tc \subseteq \Gc_d$.  In order to achieve $\pebar \le \delta$, it is necessary that
    \begin{equation}
        1 \ge \frac{\log|\Tc|}{ \sum_{i=1}^N I(G; X^{(i)} | Z^{(i)}) } \bigg(1-\delta - \frac{\log 2}{\log|\Tc|}\bigg), \label{eq:Fano}
    \end{equation}
    where $N$ is the maximum possible number of active learning rounds.  Moreover, if there exists a $p$-dimensional joint distribution $Q$ such that $D(P_{G(z)} \| Q_{(z)}) \le \epsilon(z)$ for all $G \in \Tc$ and $z \in \{0,1\}^p$, where $\epsilon(z)$ is some non-negative function, then we have
    \begin{equation}
         I(G; X^{(i)} | Z^{(i)}) \le \EE\big[ \epsilon(Z^{(i)}) \big] \label{eq:mi_bound}
    \end{equation}
    for all $i$.
\end{lem}

The proof is given in the supplementary material.  The high-level steps are as follows: (i) Bound the error probability in terms of $I(G;\Xv)$ using Fano's inequality; (ii) Use the chain rule to write $I(G;\Xv) = \sum_{i=1}^N I(X^{(i)}; G \,|\, X^{(1)},\dotsc,X^{(i-1)})$; (iii) Upper bound the summands via analogous steps to the proof of the channel coding theorem with feedback \cite[Sec.~7.12]{Cov01}; (iv) Relate the divergence $D(P_{G(z)} \| Q_{(z)})$ to $I(G; X^{(i)} | Z^{(i)})$ using similar steps to \cite{Sha14}.

\subsection{Proof of Theorem \ref{thm:Ising} (Ising model)} \label{sec:PROOF_ISING}

\subsubsection{First Bound for the Ising Model} \label{sec:IsingPf1}

We use the following ensemble in which every node has degree one.

\smallskip
\noindent\fbox{
    \parbox{0.95\columnwidth}{
        \textbf{Ensemble1 [Isolated edges ensemble]}
        \begin{itemize}
            \item Each graph in $\Tc$ consists of $\lfloor p/2 \rfloor$ node-disjoint edges that may otherwise be arbitrary.
        \end{itemize}
    }
} \bigskip

The total number of graphs is $|\Tc| = {p \choose 2}{p-2 \choose 2}\dotsc{4 \choose 2}{2 \choose 2}$ (or similarly when $p$ is an odd number), which is lower bounded by ${\lfloor p/2 \rfloor \choose 2}^{\lfloor p/2 \rfloor}$, yielding
\begin{equation}
    \log |\Tc| \ge \Big\lfloor\frac{p}{2}\Big\rfloor \log{\lfloor p/2 \rfloor \choose 2} = (p \log p)(1+o(1)). \label{eq:IsingCardinality1}
\end{equation}
To obtain a mutual information bound of the form \eqref{eq:mi_bound}, we choose $Q = P_{G'}$ with $G'$ being the empty graph, and note that for a fixed $z \in \{0,1\}^p$ containing $n(z)$ ones, $G(z)$ consists of at most $n(z)/2$ node-disjoint edges.  Since the divergence corresponding to graphs differing in a single edge is upper bounded by $\lambda \tanh \lambda$ \cite{Sha14}, and since the divergence is additive for independent products, we obtain $D(P_{G(z)} \| P_{G'(z)}) \le \frac{n(z)}{2} \lambda\tanh\lambda$, and hence \eqref{eq:mi_bound} becomes
\begin{equation}
    I(G;X^{(i)}|Z^{(i)}) \le \frac{1}{2} \EE[ n(Z^{(i)}) ]  \lambda\tanh\lambda.
\end{equation}
Summing over $i$ and noting that $\sum_{i=1}^N n(Z^{(i)}) \le n$ with probability one, since the algorithm can only use up to $n$ node observations, we obtain
\begin{equation}
    \sum_{i=1}^N I(G; X^{(i)} | Z^{(i)}) \le \frac{n}{2}\lambda\tanh\lambda.
\end{equation}
Substitution into \eqref{eq:Fano} yields the necessary condition
\begin{equation}
    n \ge \frac{2p\log p}{ \lambda \tanh \lambda } (1-\delta-o(1)), \label{eq:IsingBound1}
\end{equation}
where the numerator arises from \eqref{eq:IsingCardinality1}.

\subsubsection{Second Bound for the Ising Model} \label{sec:IsingPf2}

We use the following ensemble from \cite{San12}.

\smallskip
\noindent\fbox{
    \parbox{0.95\columnwidth}{
        \textbf{Ensemble2($m$) [Clique-minus-one ensemble]:}
        \begin{itemize}
            \item Form $\lfloor \frac{p}{m} \rfloor$ arbitrary node-disjoint cliques containing $m$ nodes each, to form a base graph $G'$.
            \item Each graph in $\Tc$ is obtained by removing a single edge from $G'$.
        \end{itemize}
    }
} \bigskip

We choose $m=d+1$, so that the maximal degree is $d$.  The total number of graphs is $\lfloor \frac{p}{m} \rfloor {m \choose 2}$, which yields
\begin{equation}
    \log |\Tc| = (\log(pd))(1+o(1)). \label{eq:IsingCardinality2}
\end{equation}
We obtain a bound of the form \eqref{eq:mi_bound} by choosing $Q = P_{G'}$ with $G'$ as in the ensemble definition.  The divergence associated with the full graphs satisfies $D(P_{G}\|P_{G'}) \le \frac{4\lambda d e^{\lambda}}{ e^{\lambda d} }$ when $\lambda d \ge 1$ \cite[Lemma 2]{San12}.  Since $G(z)$ and $G'(z)$ are common subgraphs of $G$ and $G'$, we trivially have $D(P_{G(z)}\|P_{G'(z)}) \le D(P_{G}\|P_{G'})$, and hence $D(P_{G(z)}\|P_{G'(z)})$ satisfies the same upper bound as $D(P_{G}\|P_{G'})$ regardless of $z$.  Hence, \eqref{eq:mi_bound} yields
\begin{equation}
    I(G;X^{(i)}|Z^{(i)}) \le \frac{4\lambda d e^{\lambda}}{ e^{\lambda d} }. 
\end{equation}
Since the node observation budget is $n$, the active learning can be done in at most $n/2$ rounds without loss of optimality (i.e., excluding trivial cases where only one node is observed), and we have
\begin{equation}
    \sum_{i=1}^N I(G;X^{(i)}|Z^{(i)}) \le \frac{2n\lambda d e^{\lambda}}{ e^{\lambda d} }.
\end{equation}
Substitution into \eqref{eq:Fano} yields the necessary condition
\begin{equation}
    n \ge \frac{e^{\lambda d} \log(pd) }{ 2 \lambda d e^{\lambda}} (1-\delta-o(1)), \label{eq:n_final_Ens2}
\end{equation}
where the numerator arises from \eqref{eq:IsingCardinality2}.

\subsubsection{Third Bound for the Ising Model} \label{sec:IsingPf3}

We use the following straightforward ensemble, which was also used in \cite{San12}.

\smallskip
\noindent\fbox{
    \parbox{0.95\columnwidth}{
        \textbf{Ensemble3 [Complete ensemble]:}
        \begin{itemize}
            \item $\Tc$ contains all graphs with maximal degree at most $d$, i.e., $\Tc = \Gc_{d}$.
        \end{itemize}
    }
} \bigskip

It was shown in \cite{San12} that $\log|\Tc| \ge \frac{dp}{4} \log\frac{p}{8d}$.  To bound the mutual information in \eqref{eq:Fano}, we note that the following holds when $z^{(i)}$ contains $n(z^{(i)})$ ones, and hence $n(z^{(i)})$ nodes are observed in the $i$-th round:
\begin{equation}
    I(G; X^{(i)} | Z^{(i)} = z^{(i)}) \le n(z^{(i)}) \log 2.  \label{eq:IsingBound3}
\end{equation}
This is because the remaining $p - n(z^{(i)})$ nodes are deterministically equal to $\missing$, whereas the $n(z^{(i)})$ nodes are binary and hence reveal at most $\log 2$ bits of information each.  Summing \eqref{eq:IsingBound3} over $i$ and averaging over $Z^{(i)}$, we obtain
\begin{equation}
    \sum_{i=1}^N I(G;X^{(i)}|Z^{(i)}) \le n \log 2,
\end{equation}
and substitution into \eqref{eq:Fano} yields the desired result.

\subsection{Proof of Theorem \ref{thm:Gaussian} (Gaussian model)} \label{sec:PROOF_GAUSSIAN}

\subsubsection{First Bound for the Gaussian Model} \label{sec:GaussianPf1}

We re-use Ensemble 1 above and apply the same analysis, with the only difference being the bounding of the divergence $D(P_{G_1} \| P_{G_0})$ when $G_1$ contains one edge and $G_0$ contains no edges.

When an edge is present, we let the resulting $2 \times 2$ covariance matrix and its inverse be given by
\begin{align} 
    \bSigma_{1} & = (1-\tau^2) \left[\begin{array}{ccc}
        1 & \tau  \\
        \tau & 1  \\
    \end{array}\right], \quad
    \bSigma_{1}^{-1} = \left[\begin{array}{ccc}
        1 & -\tau \\
        -\tau & 1 \\
    \end{array}\right],
\end{align}
whereas for the graph without the edge we simply have $\bSigma_0 = \bSigma_0^{-1} = \Iv$.  Both of these choices are clearly consistent with \eqref{eq:Theta_entries}.

The divergence between two zero-mean Gaussian vectors of dimension $k$ is
\begin{equation}
    D(P_1 \| P_0) = \frac{1}{2}\bigg( \Tr(\bSigma_0^{-1}\bSigma_1) - k + \log\frac{\det \bSigma_0}{ \det\bSigma_1 } \bigg), \label{eq:gaussian_div}
\end{equation}
and with the above covariance matrices and $k=2$, this simplifies to
\begin{equation}
    D(P_1 \| P_0) = \frac{1}{2}\log\frac{1}{1-\tau^2}.
\end{equation}
Hence, in analogy with \eqref{eq:IsingBound1}, we obtain
\begin{equation}
    n \ge \frac{4 p\log p}{ \log\frac{1}{1-\tau^2} } (1-\delta-o(1)). \label{eq:GaussianBound1}
\end{equation}

\subsubsection{Second Bound for the Gaussian Model} \label{sec:GaussianPf2}

We make use of the following ensemble that is similar to one in \cite{Wan10}, but with multiple cliques as opposed to only a single one.  It can also be thought of as a generalization of Ensemble 1, which corresponds to $m = 2$.

\smallskip
\noindent\fbox{
    \parbox{0.95\columnwidth}{
        \textbf{Ensemble4($m$) [Disjoint cliques ensemble]:}
        \begin{itemize}
            \item Each graph in $\Tc$ consists of $\lfloor \frac{p}{m} \rfloor$ disjoint cliques of $m$ nodes that may otherwise be arbitrary.
        \end{itemize}
    }
} \bigskip

The total number of graphs is ${p \choose m}{p-m \choose m}\dotsc{2m \choose m}{m \choose m}$ (or analogously when $p$ does not divide $m$), which is lower bounded by ${\lfloor p/2 \rfloor \choose m}^{\frac{1}{2}\lfloor \frac{p}{m} \rfloor}$, yielding
\begin{equation}
    \log |\Tc| \ge \frac{1}{2}\Big\lfloor \frac{p}{m} \Big\rfloor\log{ \lfloor p/2 \rfloor \choose m} = \bigg(\frac{p}{2}\log\frac{p}{m}\bigg)(1+o(1))
\end{equation}
assuming that $m = o(p)$ and hence $\log{ \lfloor p/2 \rfloor \choose m} = \big(m\log\frac{p}{m}\big)(1+o(1))$.  We choose $m = d+1$ so that the maximal degree is $d$, yielding
\begin{equation}
    \log |\Tc| \ge  \bigg(\frac{p}{2}\log\frac{p}{d}\bigg)(1+o(1)). \label{eq:GaussCardinality2}
\end{equation}

As in \cite{Wan10}, we let the inverse covariance matrix associated with a single clique be given by
\begin{equation}
        \bSigma_{1}^{-1} = \left[\begin{array}{ccccc}
                1+a & a & \cdots & a  \\
                a & 1+ a & \cdots & a \\
                \vdots & \vdots & \ddots & \vdots  \\
                a & a & \cdots & 1+a \\
        \end{array}\right],
\end{equation} 
for $a > 0 $, yielding a covariance matrix given by
\begin{align} 
        &\bSigma_{1} = \frac{1}{1+ma} \nonumber \\ 
        &\hspace*{-1ex}\times \hspace*{-1ex} \left[\begin{array}{ccccc}
            1+(m-1)a & -a & \cdots & -a  \\
            -a & 1+ (m-1)a & \cdots & -a \\ 
            \vdots & \vdots & \ddots & \vdots  \\
            -a & -a & \cdots & 1+(m-1)a \\
        \end{array}\right]
\end{align} 
We set $a = \frac{\tau}{1 - \tau}$ to ensure that the ratio of off-diagonals to diagonals in $\bSigma_1^{-1}$ is $\tau$, in accordance with \eqref{eq:Theta_entries}.  Note that this form of the inverse covariance matrix is slightly different to that in \eqref{eq:Theta_entries}, but the difference only amounts to scaling all observations by a factor of $\sqrt{1+a}$, and hence the recovery problem is unchanged regardless of which form is assumed. 

To obtain a bound of the form \eqref{eq:mi_bound}, we let $Q$ be jointly Gaussian with mean zero and identity covariance matrix, defining $\bSigma_0 = \bSigma_0^{-1} = \Iv$ accordingly.  We first study the behavior of the divergence $D(P_{G(z)} \| Q_{(z)})$ when all of the non-zero values of $z$ correspond to nodes within a single clique in $G$.  Hence, $z$ contains $\mtil \in \{1,\dotsc,m\}$ non-zero entries.

Letting $\bSigmatil_1$ denote an arbitrary sub-matrix of $\bSigma_1$ corresponding to $\mtil \in \{1,\dotsc,m\}$ nodes, a straightforward computation gives
\begin{align}
    \det\bSigmatil_1 &= \frac{1+(m-\mtil)a}{1+ma} = 1 - \frac{\mtil a}{1 + ma} \\
    \Tr(\bSigmatil_1) &= \mtil \frac{1+(m-1)a}{1+ma} = \mtil \bigg(1 - \frac{a}{1+ma} \bigg).
\end{align}
Defining $\bSigmatil_{0}$ analogously simply gives $\bSigmatil_{0} = \bSigmatil_{0}^{-1} = \Iv$, and hence \eqref{eq:gaussian_div} with $k = \mtil$ gives
\begin{equation}
    D(\Ptil_1 \| \Ptil_0) = \frac{1}{2}\Bigg(-\log\bigg(1 - \frac{\mtil a}{1 + ma}\bigg) - \frac{\mtil a}{1 + ma}\Bigg)
\end{equation}
for $\Ptil_0 \sim N(\bzero,\bSigmatil_0)$ and $\Ptil_1 \sim N(\bzero,\bSigmatil_1)$.

Suppose now that a single measurement consists of $n(z)$ nodes indexed by $z \in \{0,1\}^p$.  For a fixed graph $G \in \Tc$, this amounts to observing $\mtil_j$ nodes from each clique $j=1,\dotsc,\lfloor\frac{p}{m}\rfloor$, for some integers $\{\mtil_j\}$ such that $\sum_{j=1}^{\lfloor\frac{p}{m}\rfloor} \mtil_j = n(z)$. Since the divergence is additive for independent products, we obtain
\begin{multline}
    D(P_{G(z)} \| Q_{(z)} ) = \frac{1}{2}\Bigg(\sum_{j=1}^{\lfloor\frac{p}{m}\rfloor} -\log\bigg(1 - \frac{\mtil_j a}{1 + ma}\bigg) \\ - \frac{\mtil_j a}{1 + ma}\Bigg). \label{eq:div_guass1}
\end{multline}
To simplify the subsequent exposition, we write the summation as
\begin{equation}
    \sum_{j=1}^{\lfloor\frac{p}{m} \rfloor} \beta_j f(\beta_j), \label{eq:alpha_sum}
\end{equation}
where $\beta_j = \frac{\mtil_j a}{1+ma}$ and $f(\beta) = \frac{-\log(1-\beta) - \beta}{\beta}$.  We consider the maximization of \eqref{eq:alpha_sum} subject to $0 \le \beta_j \le \frac{ma}{1+ma}$ and $\sum_{j} \beta_j = \frac{n(z)a}{1+ma}$, where these constraints follow immediately from $0 \le \mtil_j \le m$ and $\sum_{j} \mtil_j = n(z)$.  It is easy to verify that the function $f(\beta)$ is increasing in $\beta$, and therefore, the maximal value of \eqref{eq:alpha_sum} is obtained by setting as many values of $\beta_j$ as possible to the maximum value $\frac{ma}{1+ma}$, and letting an additional value of $\beta_j$ equal the remainder (if any).  This amounts to setting as many values of $\mtil_j$ as possible to $m$, and letting an additional value of $\mtil_j$ equal the remainder.  The corresponding maximum value is
\begin{align}
    \sum_{j=1}^{\lfloor\frac{p}{m} \rfloor} \beta_j f(\beta_j) 
        &= \Big\lfloor \frac{n(z)}{m} \Big\rfloor \frac{ma}{1+ma} f\Big( \frac{ma}{1+ma} \Big) \nonumber \\
        & \qquad\qquad\quad + \frac{ra}{1+ma} f\Big( \frac{ra}{1+ma} \Big) \\
        &\le \frac{n(z)}{m} \frac{ma}{1+ma} f\Big( \frac{ma}{1+ma} \Big), \label{eq:sum_bound}
\end{align}
where $r$ denotes the remainder value (i.e., the additional value of $\mtil_j$ mentioned above), and \eqref{eq:sum_bound} follows by writing $f\big( \frac{ra}{1+ma} \big) \le f\big( \frac{ma}{1+ma} \big)$ using the above-mentioned monotonicity of $f$.

Roughly speaking, we have argued that given a budget of $n(z)$ nodes to observe, the ones that yield a graph that is ``furthest'' from the empty graph are those that correspond to $\lfloor \frac{n(z)}{m} \rfloor$ complete $m$-cliques, with any remainder also concentrated within a single clique.  Intuitively, this is because taking measurements from a variety of different cliques yields more independent nodes, thus being closer to the behavior of the empty graph in which all nodes are independent.

Upper bounding the summation on the right-hand side of \eqref{eq:div_guass1} by the maximum value \eqref{eq:sum_bound}, we obtain
\begin{multline}
    D(P_{G(z)} \| Q_{(z)} ) \\ \le \frac{n(z)}{2m} \Bigg( -\log\bigg(1 - \frac{m  a}{1 + ma}\bigg) - \frac{m a}{1 + ma}\Bigg). \label{eq:div_gauss2}
\end{multline}
Applying the inequality $-\log(1-\frac{\beta}{1+\beta}) - \frac{\beta}{1+\beta} \le \frac{1}{2}\log(1+\beta^2)$, we can weaken \eqref{eq:div_gauss2} to
\begin{equation}
    D(P_{G(z)} \| Q_{(z)} ) \le \frac{n(z)}{4m} \log\big(1 + (ma)^2\big). \label{eq:div_gauss2a}
\end{equation}
We obtain from \eqref{eq:div_gauss2a} and \eqref{eq:mi_bound} that 
\begin{equation}
    I(G;X^{(i)}|Z^{(i)}) \le \frac{\EE[ n(Z^{(i)}) ]}{4m} \log\big(1 + (ma)^2\big),
\end{equation}
and summing over $i$ and again noting that $\sum_{i=1}^N n(Z^{(i)}) \le n$ with probability one, we obtain
\begin{equation}
    \sum_{i=1}^N I(G; X^{(i)} | Z^{(i)}) \le \frac{n}{4m} \log\big(1 + (ma)^2\big). \label{eq:Ens4_MI_final}
\end{equation}
Substitution into \eqref{eq:Fano} yields the necessary condition
\begin{equation}
    n \ge \frac{2pd\log \frac{p}{d}}{ \log\Big(1 + \big((d+1)\frac{\tau}{1-\tau}\big)^2\Big) } (1-\delta-o(1)), \label{eq:GaussianBound2}
\end{equation}
where the numerator arises from \eqref{eq:GaussCardinality2}, and we have set $m = d+1$ and $a = \frac{\tau}{1-\tau}$.

\section{Discussion: Average Degree  vs.~Maximal Degree} \label{sec:DISCUSSION}

The question of whether the maximal degree $\dmax$ or average degree $\davg$ dictates the performance of active graphical model selection was raised in \cite{Das16},\footnote{More precisely, \cite{Das16} considers the quantity $\dmaxbar$ defined in Section \ref{sec:related}, but this coincides with $\davg$ for all ensembles considered in this paper, at least up to a multiplicative $1+o(1)$ term.} where it was suggested that it is the latter in the Gaussian case if $\tau$ is bounded away from zero.  Our results are proved by considering restricted ensembles for which $\davg = \dmax(1+o(1))$, and hence it is not immediately clear which is more fundamental.  We proceed by discussing the two for both Ising models and Gaussian models.  

We first remark that the first terms in each of \eqref{eq:Ising_final} and \eqref{eq:Gaussian_final} do not contain $d$, and they were proved by considering an ensemble where every node has degree exactly one.  Moreover, the third term in \eqref{eq:Ising_final} is trivially obtained by counting the number of graphs with maximal degree $d$, without any further restrictions, and it is unclear how to adapt this to gain insight on the role of the average degree.  Hence, to provide a distinction between $\dmax$ and $\davg$, we focus only on the second terms in \eqref{eq:Ising_final} and \eqref{eq:Gaussian_final}.

For the Ising model, the second term in \eqref{eq:Ising_final} was obtained by considering $\lfloor \frac{p}{d+1} \rfloor$ cliques of size $d+1$, and considering graphs obtained by subsequently removing a single edge, \emph{cf.}, Section \ref{sec:IsingPf2}.  In the supplementary material, we describe an analogous ensemble in which these cliques have different sizes, and show that the term $e^{\lambda \dmax}$ still arises in the resulting sample complexity bound.  Intuitively, this is because even if all cliques except the largest are known perfectly and an edge is removed from the largest one, it is still very difficult to identify that edge.   Hence, regarding this exponential term (which is the main feature of the bound), it is $\dmax$ that dictates the performance here.

For the Gaussian model, the second term in \eqref{eq:Gaussian_final} was obtained by considering graphs containing $\lfloor \frac{p}{d+1} \rfloor$ cliques of size $d+1$, \emph{cf.}, Section \ref{sec:GaussianPf2}.  In the supplementary material, we provide a natural extension of this ensemble which instead uses cliques of \emph{differing} sizes $(d_1,\dotsc,d_K)$ such that $\sum_{k=1}^K (d_k+1) = p$.  We make the mild assumption that each of these degrees behaves as $d_k = o(p)$. 

The most straightforward extension of the proof of Theorem \ref{thm:Gaussian} yields a bound of the form $n = \Omega\big( \frac{ p \dmin \log\frac{p}{\dmax} }{ \log(1+\tau\dmax) } \big)$, where $\dmin$ is the minimum degree.  This bound is rarely tight, but it can be improved by a genie argument: Reveal to the decoder all of the smallest cliques, up to a total of $(1-\alpha) p$ nodes for some $\alpha \in (0,1)$.  The decoder is left to estimate the remaining cliques among  $\alpha p$ nodes.  

In the supplementary material, we show that as long as $\alpha$ is bounded away from zero and one, this approach yields a sample complexity lower bound of the form $n = \Omega\big( \frac{ p \dmin^{(\alpha)} \log\frac{p}{\dmax} }{ \log(1+\tau\dmax) } \big)$, where $\dmin^{(\alpha)}$ is the minimum degree among the remaining $\alpha p$ nodes.  If the top $\alpha p$ node degrees in the graph coincide to within a constant factor, then we have $\dmax^{(\alpha)} = \Theta(\davg)$, and we thus match the $O((1+\davg) p \log p)$ upper bound from \cite{Das16} for fixed $\tau$, up to a logarithmic factor.

These observations support the idea proposed in \cite{Das16} that the average degree is the more fundamental quantity in the Gaussian setting with fixed $\tau$.   Note, however, that the assumptions are slightly different, due to the coherence assumption made in \cite{Das16} and the above assumption on the top $\alpha p$ node degrees. 

\section{Conclusion} \label{sec:CONCLUSION}

We have provided lower bounds on active learning for graphical model selection.  Using a variety of restricted graph ensembles, we recovered analogous bounds to those for the passive setting, suggesting that active learning does not help much in the minimax sense for the degree-bounded class $\Gc_{d}$.  Moreover, we identified an ensemble for the Ising model in which the maximal degree remains the crucial quantity, and another ensemble for the Gaussian model in which the average degree is the more important quantity.  We note that our analysis also readily extends to the edge-bounded class $\Gc_{k}$ in which all graphs have at most $k$ edges, analogously to previous works such as \cite{San12,Sca16}.


An important direction for further research is to characterize the gain (if any) that can be achieved by active learning in the case of \emph{random} graphs (e.g., Erd\"os-R\'enyi \cite{Ana12,Ana12a}, power law \cite{Tan13}), in which the maximal and average degrees can differ considerably.  Moreover, it would be of interest to understand the role of active learning when the edges have differing parameters $\{\lambda_{ij}\}$ in the Ising model, or when the values $\tau_{ij} = \frac{|\Theta_{ij}|}{\sqrt{\Theta_{ii}\Theta_{jj}}}$ differ in the Gaussian model.

\vspace*{-1.5ex}
\section*{Acknowledgment}
\vspace*{-1.5ex}

This work was supported in part by the European Commission under Grant ERC Future Proof, SNF 200021-146750 and SNF CRSII2-147633, and by the `EPFL Fellows' programme (Horizon2020 grant 665667).


%

\bibliographystyle{IEEEtran}
\bibliography{../JS_References}

\newpage
\appendix
\onecolumn

{\huge \bf Supplementary Material}

{\Large \bf ``Lower Bounds on Active Learning for Graphical Model Selection'' (Scarlett and Cevher, AISTATS 2017)}

\section{Proof of Lemma \ref{lem:Fano}}

We start with the following form of Fano's inequality \cite[Lemma 1]{Sha14}:
\begin{equation}
    1 \ge \frac{\log|\Tc|}{ I(G; \Xv) } \bigg(1-\delta - \frac{\log 2}{\log|\Tc|}\bigg),
\end{equation}
where $\Xv = (X^{(1)},\dotsc,X^{N})$.  This remains valid in the active learning setting since it only relies on the fact that $G \to \Xv \to \hat{G}$ forms a Markov chain.  Despite this common starting point, we bound the mutual information significantly differently.  Defining $\Xv^{(1,i)} = (X^{(1)},\dotsc,X^{(i)})$, we have\footnote{Here $H$ represents entropy in the discrete case (e.g., Ising), and differential entropy in the continuous case (e.g., Gaussian).}
\begingroup
\allowdisplaybreaks
\begin{align} 
    I(G;\Xv) 
        &= \sum_{i=1}^N I(X^{(i)}; G \,|\, \Xv^{(1, i-1)}) \label{eq:boudMI_1} \\
        &= \sum_{i=1}^N I(X^{(i)}; G \,|\, \Xv^{(1, i-1)},Z^{(i)}) \label{eq:boudMI_2}\\
        &=  \sum_{i=1}^N  \Big( H(X^{(i)} \,|\, \Xv^{(1, i-1)},Z^{(i)}) - H(X^{(i)} \,|\, \Xv^{(1, i-1)},Z^{(i)},G) \Big) \label{eq:boudMI_3} \\
        &=  \sum_{i=1}^N  \Big( H(X^{(i)} \,|\, \Xv^{(1, i-1)},Z^{(i)})  - H(X^{(i)} \,|\, Z^{(i)},G) \Big) \label{eq:boudMI_4} \\
        &\le  \sum_{i=1}^N  \Big( H(X^{(i)} \,|\, Z^{(i)}) - H(X^{(i)} | G,Z^{(i)}) \Big) \label{eq:boudMI_5} \\
        &= \sum_{i=1}^N I(G;X^{(i)}|Z^{(i)}), \label{eq:boudMI_6}
\end{align}
\endgroup
where \eqref{eq:boudMI_1} follows from the chain rule, \eqref{eq:boudMI_2} follows since $Z^{(i)}$ is a function of $\Xv^{(1,i-1)}$ , \eqref{eq:boudMI_4} follows since $X^{(i)}$ is conditionally independent of $\Xv^{(1,i-1)}$ given $(G,Z^{(i)})$, and \eqref{eq:boudMI_5} follows since conditioning reduces entropy.  This completes the proof of \eqref{eq:Fano}.

Conditioned on $Z^{(i)} = z^{(i)}$, the only variables in $X^{(i)}$ conveying information about $G$ are those corresponding to entries where $z^{(i)}$ is one, since the others deterministically equal $\missing$.  By applying the mutual information upper bound of \cite{Sha14} (see the proof of Corollary 2 therein) to the restricted graph $G(z^{(i)})$ with an auxiliary distribution $Q_{(z^{(i)})}$, we obtain that
\begin{equation}
    D(P_{G(z^{(i)})} \| Q_{(z^{(i)})}) \le \epsilon(z^{(i)}), \forall G\in\Tc \implies I(G;X^{(i)}|Z^{(i)}=z^{(i)}) \le \epsilon(z^{(i)}). \label{eq:D_I_init}
\end{equation}
Note that conditioned on $Z^{(i)}=z^{(i)}$, the graph $G$ may no longer be uniform on $\Tc$; the preceding claim remains valid since the proof of \cite[Cor.~2]{Sha14} is for general graph distributions that need not be uniform.

Finally, the inequality in \eqref{eq:mi_bound} follows by averaging both sides of the mutual information bound in \eqref{eq:D_I_init} over $Z^{(i)}$.

\section{Ensemble and Sample Complexity for Comparing the Average Degree and Maximal Degree (Ising model)}

Formalizing the discussion on the Ising model in Section \ref{sec:DISCUSSION}, we introduce the following analog of Ensemble 2, consisting of some number $L$ of variable-size cliques with an edge removed.

\smallskip
\noindent\fbox{
    \parbox{0.95\columnwidth}{
        \textbf{Ensemble2a($m_1,\dotsc,m_L$) [Variable-size edge-removed cliques ensemble]:}
        \begin{itemize}
            \item Form $L$ arbitrary node-disjoint cliques of sizes $(m_1,\dotsc,m_L)$, to obtain a base graph $G'$.
            \item Each graph in $\Tc$ is obtained by removing a single edge from each of the $L$ cliques.
        \end{itemize}
    }
} \bigskip

We have the following.

\begin{lem}
    Fix the integers $L$ and $(m_1,\dotsc,m_L)$ with $\sum_{j=1}^L m_j = p$, and let $G$ be drawn uniformly from \emph{Ensemble2a}($m_1,\dotsc,m_L$).  Then in order to achieve $\pebar \le \delta$, it is necessary that
    \begin{equation}
        n \ge \frac{e^{\lambda\dmax} \log\big( \dmax(\dmax+1) \big)}{2 \lambda \dmax e^{\lambda}} \Big(1 - \delta - \frac{\log 2}{\log(\dmax+1)}\Big),
    \end{equation}
    where $\dmax = \max_{j=1,\dotsc,L} m_j - 1$.
\end{lem}
\begin{proof}
    We consider a genie argument, in which the decoder is informed of all of the removed edges from the cliques, except for the largest, whose size is $\dmax + 1$.  In this case, the analysis reduces to that of Ensemble2($\dmax + 1$) on a graph with $p = \dmax + 1$ nodes.  The result now follows immediately from \eqref{eq:n_final_Ens2}, and recalling that the $o(1)$ remainder term therein is equal to $\frac{\log 2}{|\Tc|}$ from \eqref{eq:Fano}.
\end{proof}

\section{Ensemble and Sample Complexity for Comparing the Average Degree and Maximal Degree (Gaussian model)}

Formalizing the discussion on the Gaussian model in Section \ref{sec:DISCUSSION}, we introduce the following ensemble, consisting of some number $L$ of variable-size cliques.

\smallskip
\noindent\fbox{
    \parbox{0.95\columnwidth}{
        \textbf{Ensemble4a($m_1,\dotsc,m_L$) [Disjoint variable-size cliques ensemble]:}
        \begin{itemize}
            \item Each graph in $\Tc$ consists of $L$ disjoint cliques of sizes $(m_1,\dotsc,m_L)$ nodes that may otherwise be arbitrary.
        \end{itemize}
    }
} \bigskip

We have the following.

\begin{lem}
    Fix the integers $L$ and $(m_1,\dotsc,m_L)$ with $\sum_{j=1}^L m_j = p$ and $\max_{j=1,\dotsc,L} m_j = o(p)$, and let $G$ be drawn uniformly from \emph{Ensemble4a}($m_1,\dotsc,m_L$).  Then for any $\alpha \in (0,1)$ (not depending on $p$), in order to achieve $\pebar \le \delta$, it is necessary that
    \begin{equation}
        n \ge\frac{2\alpha p \dmin^{(\alpha)}\log \frac{p}{\dmax}}{ \log\Big(1 + \big((\dmax+1)\frac{\tau}{1-\tau}\big)^2\Big) } \big(1 - \delta - o(1)\big),
    \end{equation}
    where $\dmax = \max_{j=1,\dotsc,L} m_j - 1$, and $\dmin^{(\alpha)}$ is the minimum degree among the $\alpha p$ nodes having the largest degree.\footnote{This is the same for all graphs in the ensemble, so here $\dmin^{(\alpha)}$ is well-defined.}
\end{lem}
\begin{proof}
    We again consider a genie argument, in which the decoder is informed of all of the cliques except the largest ones, such that these remaining cliques form a total of $\alpha p$ nodes.\footnote{Since $m_j = o(p)$ for all $j$, we can safely ignore rounding and assume that the total is exactly $\alpha p$.}  Assuming without loss of generality that the $m_j$ are in decreasing order, the analysis reduces to the study of Ensemble4a on a graph with $\alpha p$ nodes, and cliques of size $(m_1,\dotsc,m_{L'})$, where $L' \le L$ is defined such that $\sum_{j=1}^{L'} m_j = \alpha p$.
    
    For this reduced ensemble, the total number of graphs is ${\alpha p \choose m_1}{\alpha p-m_1 \choose m_2}\dotsc{\alpha p - \sum_{j=1}^{L'-2} m_j \choose m_{L' - 1} }{m_{L'} \choose m_{L'}}$.  We let $L''$ be the largest integer such that $\sum_{j=1}^{L''} m_j \le \alpha p / 2$, and write
    \begin{align}
        \log |\Tc| 
            &\ge \sum_{j=1}^{L''} \log{ \lfloor \alpha p/2 \rfloor \choose m_j} \label{eq:Cardinality_a1}\\
            &= \sum_{j=1}^{L''}  \Big( m_j \log \frac{\alpha p}{2 m_j} \Big) (1+o(1)) \label{eq:Cardinality_a2} \\
            &\ge \Big( \frac{\alpha p}{2} \log \frac{\alpha p}{2 m_1} \Big) (1+o(1)) \label{eq:Cardinality_a3} \\
            &= \Big( \frac{\alpha p}{2} \log \frac{p}{\dmax} \Big) (1+o(1)), \label{eq:Cardinality_a4}
    \end{align}
    where \eqref{eq:Cardinality_a2} follows since $m_j = o(\alpha p)$ by assumption, \eqref{eq:Cardinality_a3} follows by first applying $m_j \le m_1$ inside the logarithm and then applying the definition of $L''$, and \eqref{eq:Cardinality_a4} follows since $m_1 = \dmax + 1$ by definition.
    
    
    We now follow the analysis of Section \ref{sec:GaussianPf2}, and note that if a single measurement consists of $n(z)$ nodes indexed by $z \in \{0,1\}^p$, and if this corresponds to observing $\mtil_j$ nodes from each clique $j=1,\dotsc,L'$, then we have the following analog of \eqref{eq:div_guass1}:
    \begin{equation}
        D(P_{G(z)} \| Q_{(z)} ) = \frac{1}{2}\Bigg(\sum_{j=1}^{L'} -\log\bigg(1 - \frac{\mtil_j a}{1 + m_j a}\bigg) - \frac{\mtil_j a}{1 + m_j a}\Bigg), \label{eq:div_guass1a}
    \end{equation}
    where $Q_{(z)}$ and $a$ are defined in Section \ref{sec:GaussianPf2}.
    
     Defining $\beta_j = \frac{\mtil_j a}{1+m_j a}$ and $f(\beta) = \frac{-\log(1-\beta) - \beta}{\beta}$, we can write the right-hand side of \eqref{eq:div_guass1a} as
    \begin{equation}
        \sum_{j=1}^{L'} \beta_j f(\beta_j), \label{eq:alpha_sum_a}
    \end{equation}
    As a result, we consider the maximization of \eqref{eq:alpha_sum} subject to $0 \le \beta_j \le \frac{m_ja}{1+m_ja}$ and $\sum_{j} \beta_j (1+m_j a) = n(z)a$, where these constraints follow immediately from $0 \le \mtil_j \le m$ and $\sum_{j} \mtil_j = n(z)$. 
    
    While the optimal choices of $\{\beta_j\}$ for the preceding maximization problem are unclear, we observe that the final objective value can only increase if we relax the second constraint to $\sum_{j} \beta_j (1+\mmin^{(\alpha)} a) \le n(z)a$, where $\mmin^{(\alpha)} = m_{L'} = \dmin^{(\alpha)} + 1$.  With this modification, we find similarly to \eqref{eq:div_guass1} that the maximum is achieved by setting $\beta_j$ to its maximum value $\frac{m_j a}{1 + m_j a}$ (i.e., $\mtil_j = m_j$) for as many of the largest cliques as is permitted by the constraint $\sum_{j} \beta_j (1+\mmin^{(\alpha)} a) \le n(z)a$.  Since each clique under consideration has at least $\mmin^{(\alpha)}$ nodes, this amounts to at most $\frac{n(z)}{\mmin^{(\alpha)}}$ cliques.  Moreover, since $\beta f(\beta)$ is increasing in $\beta$, the corresponding values of $\beta_j f(\beta_j)$ are upper bounded by $\frac{\mmax a}{1 + \mmax a}  f\big( \frac{\mmax a}{1 + \mmax a} \big)$.
    
    Combining these observations, we obtain the following analog of \eqref{eq:sum_bound}:
    \begin{align}
        \sum_{j=1}^{L'} \beta_j f(\beta_j) \le \frac{n(z)}{\mmin^{(\alpha)}} \frac{\mmax a}{1+\mmax a} f\Big( \frac{\mmax a}{1+\mmax a} \Big), \label{eq:sum_bound_a}
    \end{align}
    and accordingly, using the same subsequent steps, we obtain the following analog of \eqref{eq:Ens4_MI_final}:
    \begin{equation}
        \sum_{i=1}^N I(G; X^{(i)} | Z^{(i)}) \le \frac{n}{4\mmin^{(\alpha)}} \log\big(1 + (\mmax a)^2\big). \label{eq:Ens4_MI_final_a}
    \end{equation}
    The proof is concluded using \eqref{eq:Fano} along with the cardinality bound in \eqref{eq:Cardinality_a4}, and recalling that $\mmin^{(\alpha)} = \dmin^{(\alpha)} + 1$, $\mmax = \dmax + 1$, and $a = \frac{\tau}{1-\tau}$.
    
\end{proof}

\end{document}